%% file: ms.tex
\documentclass{vldb}
\usepackage{graphicx}
\usepackage{balance}  
\input{defs}

\vldbTitle{Cross-chain Deals and Adversarial Commerce}
\vldbAuthors{Maurice Herlihy, Barbara Liskov, and Liuba Shrira}
\vldbDOI{https://doi.org/10.14778/3364324.3364326}
\vldbVolume{13}
\vldbNumber{2}
\vldbYear{2019}

\begin{document}
\title{Cross-chain Deals and Adversarial Commerce}
\numberofauthors{3}
 \author{
\alignauthor
Maurice Herlihy\\
\affaddr{Brown University}\\
\email{mph@cs.brown.edu}
\alignauthor
Barbara Liskov\\
\affaddr{MIT CSAIL}\\
\email{liskov@csail.mit.edu}
\alignauthor
Liuba Shrira \\ 
\affaddr{Brandeis University}\\
\email{liuba@brandeis.edu}
}

\maketitle

\begin{abstract}
\input{abstract}
\end{abstract}

\section{Introduction}
\input{intro}

\section{Cross-Chain Deals}
\input{protocol}

\section{System Model}
\input{model}

\section{How Cross-Chain Deals Work}
\input{deal}

 \section{Timelock Protocol}
 \seclabel{hash}
 \input{timelock}

\section{CBC Protocol}
 \seclabel{cbc}
  \input{cbc}

\section{Cost Analysis}
\input{costs}

\section{Related Work}
\seclabel{related}
\input{related}

\section{Discussion}
\input{discussion}

\section{Conclusions}
\input{conclusions}
\balance
\newpage
\bibliographystyle{abbrv}
\bibliography{blockchain}

\end{document}

%% file: defs.tex
\usepackage{listings}
\usepackage[hyphens]{url}
\newcommand{\var}[1]{\text{\lstinline+#1+}}
\newtheorem{theorem}{Theorem}
\newtheorem{property}{Property}

\def\cA{\ensuremath{{\mathcal A}}}

         \def\cP{\ensuremath{{\mathcal P}}}

\newcommand{\Owns}{\ensuremath{\mathit{Owns}}}
\newcommand{\OwnsC}{\ensuremath{\mathit{Owns}_C}}
\newcommand{\OwnsA}{\ensuremath{\mathit{Owns}_A}}
\newcommand{\plist}{\ensuremath{\mathit{plist}}}
\newcommand{\Dinfo}{\ensuremath{\mathit{Dinfo}}}

\newcommand{\escrow}{\ensuremath{\mathit{escrow}}}

\newcommand{\transfer}{\ensuremath{\mathit{transfer}}}
\newcommand{\startX}{\ensuremath{\mathit{startDeal}}}
\newcommand{\startD}{\ensuremath{\mathit{startDeal}}}
\newcommand{\startDeal}{\startD}
\newcommand{\commit}{\ensuremath{\mathit{commit}}}
\newcommand{\abort}{\ensuremath{\mathit{abort}}}

\usepackage{epsfig}

\newcommand{\figlabel}[1]{\label{figure:#1}}
\newcommand{\nakedfigref}[1]{\ref{figure:#1}}
\newcommand{\figref}[1]{Figure~\nakedfigref{#1}}
\newcommand{\Figref}[1]{Figure~\nakedfigref{#1}}

\newcommand{\tablabel}[1]{\label{table:#1}}
\newcommand{\nakedtabref}[1]{\ref{table:#1}}
\newcommand{\tabref}[1]{Table~\nakedtabref{#1}}

\newcommand{\seclabel}[1]{\label{sec:#1}}
\newcommand{\nakedsecref}[1]{\ref{sec:#1}}
\newcommand{\secref}[1]{Section~\nakedsecref{#1}}

\newcommand{\linelabel}[1]{\label{line:#1}}
\newcommand{\nakedlineref}[1]{\ref{line:#1}}
\newcommand{\lineref}[1]{Line~\nakedlineref{#1}}

\newcommand{\proplabel}[1]{\label{prop:#1}}


%% file: abstract.tex
Modern distributed data management systems face a new challenge:
how can autonomous, mutually-distrusting parties cooperate
safely and effectively?
Addressing this challenge brings up questions familiar from
classical distributed systems:
how to combine multiple steps into a single atomic action,
how to recover from failures, and
how to synchronize concurrent access to data.
Nevertheless, each of these issues requires rethinking when
participants are autonomous and potentially adversarial.

We propose the notion of a \emph{cross-chain deal},
a new way to structure complex distributed computations that
manage assets in an adversarial setting.
Deals are inspired by classical atomic transactions,
but are necessarily different, in important ways,
to accommodate the decentralized and untrusting nature of the exchange.
We describe novel safety and liveness properties,
along with two alternative protocols for implementing
cross-chain deals in a system of independent blockchain ledgers.
One protocol, based on synchronous communication,
is fully decentralized,
while the other,
based on semi-synchronous communication,
requires a globally shared ledger.

%% file: intro.tex
The emerging domain of electronic commerce spanning multiple
blockchains is a kind of fun-house mirror of classical distributed computing:
familiar features are recognizable, but distorted.
For example,
classical atomic transactions are often described in terms of the
well-known ACID properties~\cite{HaerderR1983}:
atomicity, consistency, isolation, and durability.
We will see that cross-chain commerce requires structures
superficially similar to, but fundamentally different from,
classical atomic transactions.
In particular,
the classical notions of correctness for atomic transactions
must be rethought.

Classical \emph{atomicity} means that a transaction's effects take
place everywhere or nowhere.
This notion of atomicity cannot be guaranteed when parties are
potentially malicious:
the best one can do is to ensure that honest parties cannot be cheated.
Moreover, classical transactions often prioritize safety over
liveness, allowing, for example, commit protocols that can
block~\cite{twophasecommit}.
For cross-chain commerce, however,
liveness guarantees have been relied upon to ensure that one party cannot trick
another into locking up assets forever, or even for a long time.

Classical \emph{isolation} guarantees that concurrent transactions
cannot interfere in destructive ways.
Isolation is typically provided via properties such as
serializability or snapshot consistency.
These properties are poorly suited to cross-chain commerce,
where mutually-untrusting parties may require multiple cautious interactions to
set up and execute a deal.
Instead,
non-interference takes the form of protection against
\emph{double-spending}: ensuring that one party does not concurrently
sell the same asset to multiple counterparties.

Here we propose the notion of a \emph{cross-chain deal}, a new
computational abstraction for structuring commercial interactions
that mirror standard (non-blockchain) commercial practices.
Cross-chain deals are inspired by classical atomic transactions and
modern cross-chain swaps,
but differ, in essential ways, from both.

Cross-chain deals are not atomic transactions.
They solve different problems:
transactions perform complex distributed state changes,
while deals, by contrast, simply exchange assets among parties.
While a transaction's effects must be ``all-or-nothing'' to preserve
global invariants,
each autonomous party in a deal can decide independently whether it
finds an outcome satisfactory for itself.
Transactions and deals make different failure assumptions: 
transactions usually assume parties can fail only by crashing,
while deals necessarily assume parties may deviate arbitrarily from
the common protocol.

Cross-chain deals are also not cross-chain swaps.
As illustrated below,
cross-chain swaps lack the expressive power to support all but the
simplest kinds of standard commercial practices.

\sloppy
Adversarial commerce, defined as economic exchange
among mutually untrusted autonomous parties, is here to stay.
Moreover, a system architecture composed of autonomous untrusted
parties that communicate via shared tamper-proof data stores is the
most natural way to organize such a system.
Although we will propose protocols based on today's blockchains and
smart contracts, none of our principal results depends on specific
blockchain technology, or even blockchains as such.
Instead, we focus on computational abstractions central to any
systematic approach to adversarial commerce,
no matter what technology underlies the shared data stores.

This paper makes the following contributions.
\begin{itemize}
\item
We propose the \emph{cross-chain deal} as a new computational
abstraction for structuring complex distributed exchanges in an
adversarial setting.
Deals require new notions of correctness and new distributed protocols.

\item
We propose new safety and liveness properties to replace the classical
notions of transactional atomicity.

\item
We describe two protocols for implementing cross-chain deals:
a fully-decentralized \emph{timelock} protocol that assumes a synchronous
communication model,
and a more centralized \emph{certified blockchain} (CBC) protocol that does not.
We sketch a proof that any protocol that tolerates periods of
asynchrony must rely on a centralized blockchain (or similar ledger structure).

\item
We sketch implementations of the two protocols and use them to analyze
the costs of the protocols.  Given the immature state of today's
blockchain technology, we focus on inherent, abiding,
platform-independent trade-offs and costs rather than explicit
performance measurements. Our intent is to illustrate the
costs associated with the
qualitatively different ways in which cross-chain deals can be
supported. 
\end{itemize}

Specifying correctness for systems in which parties can deviate
arbitrarily from the common protocol requires care.
One cannot even assume parties will be rational,
because they may have unknown objective functions
(such as a foreign power willing to pay to disrupt an economy).
Furthermore,
some familiar classical properties such as the ``all-or-nothing''
property of atomic transactions cannot be enforced in an
adversarial setting.
In much of the literature on cross-chains swaps and transactions,
correctness is treated informally, or in an ambiguous way.
Without a realizable (and realistic) notion of correctness,
it is impossible to reason about the correctness of blockchain protocols,
smart contract code,
or any other subsystem supporting adversarial commerce.

\subsection{System Models}
The cross-chain deal abstraction is flexible enough to encompass
a variety of system models and assumptions.
At one extreme,
proof-of-work protocols such as Bitcoin~\cite{bitcoin} and Ethereum~\cite{ethereum}
require a \emph{synchronous} model,
where there is a known upper bound on the propagation time for one party's change to the blockchain state to be noticed by the other parties.
At the other extreme,
Byzantine fault-tolerant consensus protocols such as Algorand~\cite{GiladHMVZ2017}, Libra~\cite{libra}, and Hot-Stuff~\cite{AbrahamGM2018} operate in a more demanding \emph{semi-synchronous} model,
where there is initially no bound on propagation time,
but the system eventually reaches a \emph{global stabilization time}
(GST) after which the system becomes synchronous.
(In practice, the synchronous periods need only last ``long enough'' to stabilize the protocol.)

Each model has advantages and disadvantages,
and each imposes its own constraints on protocols and applications.
Synchronous protocols may require countermeasures against denial-of-service attacks,
such as choosing timeouts carefully, or establishing ``watchtowers''~\cite{watchtower}.
Semi-synchronous protocols avoid explicit timeouts,
but they necessarily require a greater degree of centralization than
synchronous protocols.
(Indeed, we sketch a proof that any semi-synchronous
protocol there must be a single ledger consulted by all parties.) 
Choosing a communication model involves complex application-specific trade-offs that
lie beyond the scope of this paper.

To illustrate the broad applicability of cross-chain deals,
we introduce two protocols for implementing cross-chain deals:
a fully-decentralized \emph{timelock} protocol that assumes a synchronous
communication model,
and a more centralized \emph{certified blockchain} protocol that does not.
While both protocols implement cross-chain deals,
we will see that they make different guarantees about what happens
when things go wrong.

\subsection{Other Approaches}
Today, the most common way to trade electronic assets is to use a
trusted third party,
sometimes called a \emph{crypto exchange},
such as Coinbase or Binance.
Such exchanges are typically unregulated, and provide no guarantees of any kind.
Crypto exchanges have been known to lose substantial sums to hackers~\cite{binance},
and even to vanish along with their customers' deposits~\cite{mtgox}.

The need for safety in adversarial commerce has inspired an outburst
of interest in
\emph{cross-chain swaps}~\cite{bitcoinwiki,bip199,decred,arwen2019,Herlihy2018,tiersnolan,barterdex,ZakharyAE2019,Catalyst},
supporting atomic, one-shot, unconditional direct transfers between parties.
Existing cross-chain swap proposals have very limited expressive power:
they cannot support indirect transfers, as mediated by a broker,
nor can they support conditional exchanges such as auctions,
where the seller exchanges assets only with the highest bidder.
In general, cross-chain swaps do not support the kinds of complex ``business logic'' 
required by many kinds of modern financial deals.

\subsection{Example}
Alice is a ticket broker.  She buys tickets at wholesale prices from
event organizers and resells them at retail prices to consumers,
collecting a modest commission.  Alice lives in the future, where
tickets are managed on a \emph{ticket blockchain}, a tamper-proof
replicated ledger that tracks ticket ownership.  Similarly, the coins
paid for those tickets live on a distinct \emph{coin blockchain}.
Both blockchains support \emph{contracts}, simple programs that
control when and how ownership of tickets and coins is transferred.

One day Bob, a theater owner,
decides to sell two coveted tickets to a hit play for 100 coins.
Alice knows that Carol would be willing to pay 101 coins for those tickets,
so Alice moves to broker a deal between Bob and Carol.

Alice's task is to devise a distributed protocol, executed by Alice,
Bob, and Carol, communicating through contracts running on various
blockchains, to execute a cross-chain deal that transfers the tickets
from Bob to Carol, and the coins from Carol to Bob, minus Alice's
commission.  If all goes as planned, all transfers take place, and if
anything goes wrong (someone crashes or tries to cheat), no honest
party should end up worse off.  For example, Alice should not
end up holding tickets she can't sell or coins that she must refund.
\newpage
\subsection{Cross-Chain Deals vs. Cross-chain Swaps}
Existing proposals for cross-chain swaps set up a collection of
unconditional transfers from one party to another.
Each party checks that it is satisfied with its own transfers in and out,
and then when all approve, the transfers take place.
In most of these proposals,
transfers are triggered by producing the preimage to a hashed
value within a certain time (a \emph{hashed timelock}),
but there are proposals to use blockchain-based two-phase commit protocols~\cite{ZakharyAE2019}.

Because each party sets up an unconditional transfer,
cross-chain swaps are incapable of expressing many kinds of standard financial transactions.
For example:
\begin{itemize}
\item
  Even Alice, Bob, and Carol's simple ticket brokering deal
  cannot be expressed as a cross-chain swap.
  Alice pays Bob with coins she receives from Carol,
  and sells Carol a ticket she receives from Bob.
  Alice cannot commit to either transfer at the start of the protocol
  because she does not own those assets.
Instead, Alice acts as a broker,
contributing value by acting as middleman,
relaying assets between Bob and Carol
(or more generally,
between a pool of wholesale sellers and a pool of retail buyers).

\item
  As part of the ticket deal,
  Carol wants to convert alt-coins into coins,
  which she will use to pay Alice.
  If the deal does not go through,
  she does not want to be left holding extra coins.
  Carol cannot statically commit to transfer those coins,
  because she does not own them at the start of the protocol.

\item
  Suppose instead Alice wants to auction an asset.
  Bob and Carol escrow their bids.
  Alice's contract compares those bids,
  and if the higher bid exceeds her reserve price,
  it releases the lower bid to the loser,
  and exchanges the asset for the winner's higher bid.
  Auctions cannot be expressed as swaps because the auction's outcome
  (reserve price exceeded, identity of winner)
  cannot be determined until all bids have been submitted.

\item
  More generally,
  cross-chain swaps cannot support any kind of \emph{arbitrage},
  where Alice acquires an asset at one price,
  and immediately resells it at a higher price.
\end{itemize}

We are not aware of any prior cross-chain swap protocol that can support
these (conventional) kinds of structured deals.
While cross-chain swaps can be considered a special case of cross-chain deals,
cross-chain deals are substantially more flexible and powerful.

%% file: protocol.tex
Here we describe cross-chain deals,
what it means to execute them,
and what it means for them to be correct.

\subsection{Specifying the Deal}
\begin{table}
\centering
\caption{Alice, Bob, and Carol's deal. Rows represent outgoing transfers, and columns incoming transfers}
\tablabel{matrix}

\begin{tabular}{|l|c|c|c|}
  \hline
        &Alice     &Bob        &Carol \\
  \hline
  Alice &          &100 coins  &tickets \\
  \hline
  Bob   &tickets   &           &        \\
  \hline
  Carol &101 coins &           &        \\
  \hline
\end{tabular}
\end{table}

Each payoff (set of final transfers) for a deal can be expressed as a matrix (or table),
where each row and column is labeled with a party,
and the entry at row $i$ and column $j$ shows the assets
to be transferred from party $i$ to party $j$.
A party's column states what it expects to acquire from the deal
(its \emph{incoming} assets),
and its row states what it expects to relinquish
(its \emph{outgoing} assets).
A party enters a deal if the proposed transfers leave it better off,
and it agrees to commit (complete) the deal if it deems the
actual payoff to be acceptable.

In our running example,
the payoff is given by the $3 \times 3$ matrix in \tabref{matrix}.
Carol expects to transfer 101 coins to Alice
in return for tickets transferred from Alice.
Similarly, Bob expects to transfer tickets to Alice in return for 100
coins from Alice.
Although the table refers only to ``tickets,''
the specific (non-fungible) tickets to be provided would be part of
the deal specification,
while the specific (fungible) coins would likely be omitted.

A deal where Alice auctions an asset to Bob and Carol
would require two matrices, one for each successful outcome:
(1) if Bob outbids Carol, Alice transfers her asset to Bob,
Bob transfers his bid to Alice, and Carol transfers her bid back to herself, and
(2) if Carol outbids Bob, the transfers are symmetric.
(A realistic on-chain auction would also include fees and deposits to
penalize malicious behavior by bidders.)

\subsection{Correctness}
Parties to a deal carry out a \emph{protocol}
to complete the deal's transfers.
In an environment where we cannot force the parties to follow 
a protocol,
it is impossible to guarantee that all transfers take
place as promised by the deal specification.
Which kind of partial transfers should be deemed acceptable?

Instead of distinguishing between faulty and non-faulty parties,
as in classical models,
we distinguish only between \emph{compliant} parties
who follow the protocol,
and \emph{deviating} parties who do not.
Many kinds of fault-tolerant distributed protocols require that some
fraction of the parties be compliant.
For example,
proof-of-work consensus~\cite{bitcoin} requires a compliant majority,
while most Byzantine fault-tolerant (BFT) consensus protocols require
more than two-thirds of the participants to be compliant.
For cross-chain deals, however,
it seems prudent to make no assumptions about the number of deviating parties.

This classification of parties as either compliant or deviating is
partly inspired by the classification in the \emph{BAR}
model~\cite{AiyerACDMP2005}, which identifies parties as rational,
altruistic, or Byzantine.
In adversarial commerce, compliant parties are rational parties that
choose to follow the protocol; all others are deviating, whether rational or not
(c.f. \cite{FordB2018}).
There are no altruistic parties.
Critically, our classification differs from that of BAR (and
other standard models of Byzantine behavior) by not limiting the
number of Byzantine parties.

The most fundamental safety property is (informally) that compliant
parties should end up ``no worse off,''
even when other parties deviate arbitrarily from the protocol.
A party's \emph{payoff} for a protocol execution is the sets of
incoming and outgoing assets actually transferred.
Some payoffs are considered \emph{acceptable}, the rest not.
Some acceptable payoffs are preferable to others,
but any acceptable payoff leaves that party ``no worse off.''

Every party considers the following payoffs acceptable:
\textsc{All}, where all agreed transfers take place,
and \textsc{Nothing}, where no transfers take place.
In addition,
we allow a party to consider other payoffs acceptable.
For example, a party that expects three incoming transfers
and three outgoing transfers may be willing to accept a payoff where
it receives only two incoming transfers in return for
only two outgoing transfers.
Of course, any such choice is application-dependent.

We also assume that if a payoff is acceptable to a party
then so is any payoff where that party transfers in strictly
more incoming assets (something for nothing)
or transfers out strictly fewer outgoing assets (discount pricing).
For example,
a payoff where a party transfers no outgoing assets but receives some
incoming assets is an acceptable modification to the \textsc{Nothing}
payoff. 
Such outcomes, while unlikely in practice, cannot be excluded.

A cross-chain deal protocol satisfies \emph{safety} if:
\begin{property}
  \proplabel{safety}
For every protocol execution,
every compliant party ends up with an acceptable payoff.
\end{property} 
This notion of safety replaces the classical \emph{all-or-nothing}
property of atomic transactions,
which, as noted, cannot be implemented in the presence of deviating parties.

Cross-chain task protocols typically rely on some form of
\emph{escrow} to ensure the good faith of
participating parties.  The following \emph{weak liveness} property
ensures that conforming parties' assets cannot be locked up forever.

\begin{property}
No asset belonging to a compliant party is escrowed forever.
\end{property}

Finally, we would like protocols to satisfy the following
\emph{strong liveness} property:
\begin{property}
\proplabel{strongliveness}
If all parties are compliant and willing to accept their proposed payoffs,
then all transfers happen
(all parties' payoffs are \textsc{All}).
\end{property}
It is a well-known result~\cite{FischerLP1985} that strong liveness is
possible only in periods when the communication network is
synchronous, 
ensuring a fixed upper bound on message delivery time.

%% file: model.tex
For our purposes,
a \emph{blockchain} is a publicly-readable,
tamper-proof distributed ledger (or database) that
tracks ownership of \emph{assets} among various \emph{parties}.
An asset may be \emph{fungible}, like a sum of money,
or \emph{non-fungible}, like a theater ticket.
A party can be a person, an organization, or even a contract (see below).
We assume multiple independent blockchains,
each managing a different kind of asset.
We restrict our attention to blockchains that track asset ownership,
and to deals that transfer asset ownership from one party to another.
We assume all value transfers are explicitly represented on the blockchain.
For example, Alice does not send paper tickets to Carol off-chain.

A party can \emph{publish} an entry on a blockchain,
and it can \emph{monitor} one or more blockchains,
receiving notifications when other parties publish entries.
In our model,
publishing an entry usually executes a blockchain-resident program
called a \emph{contract}.
We will use contracts for \emph{escrow}:
an asset owner temporarily transfers ownership of an asset to a contract.
If certain conditions are met,
the contract transfers that asset to a \emph{counterparty},
and otherwise it refunds that asset to the original owner.

A party can publish a new contract on a blockchain,
or call a function exported by an existing contract.
Contract code and contract state are public,
so a party calling a contract knows what code will be executed.
Contract code must be deterministic because
contracts are typically re-executed multiple times by
mutually-suspicious parties.

A contract accesses data on the blockchain where it resides,
but it cannot directly access data from the outside world,
and cannot call contracts on other blockchains.
A contract on blockchain $A$ can learn of a change to
a blockchain $B$ only if some party explicitly informs $A$ of $B$'s change,
along with some kind of ``proof'' that the information about $B$'s state is correct.

In summary,
contract code is passive, public, deterministic, and trusted,
while parties are active, autonomous, and potentially dishonest.
Parties are given a protocol, which each party may or may not follow.
Parties may or may not act rationally~\cite{AiyerACDMP2005}.

We make standard cryptographic assumptions.
Each party has a public key and a private key,
and any party's public key is known to all.
Messages are signed so they cannot be forged,
and they include single-use labels (``nonces'')
so they cannot be replayed.

%% file: deal.tex
We can model a cross-chain deal in terms of a simple state
machine that tracks ownership of assets,
and whose transitions represent escrows, transfers, commits, and aborts.

Let $\cP$ be a domain of \emph{parties}, and $\cA$ a domain of \emph{assets}.
(A party may be a person or a contract,
and assets are digital tokens representing items of value.)
An asset has exactly one \emph{owner} at a time:
$\Owns(P,a)$ is \emph{true}
if $P$ and only $P$ owns $a$.

An active deal tentatively transfers asset ownership from one party to another.
We say a tentative transfer \emph{commits} if it becomes permanent,
and it \emph{aborts} if it is discarded.
A deal \emph{commits} if all its tentative transfers commit,
and it \emph{aborts} if all its tentative transfers abort.

While a deal is in progress, its state encompasses two maps,
$C: \cA \to \cP$ and $A: \cA \to \cP$, both initially empty.
$C(a)$ indicates the eventual owner of asset $a$ if the deal commits at $a$'s blockchain,
and $A(a)$ the owner if it aborts at that blockchain.
We use $\OwnsC(P,a)$ to indicate that $P$ will own $a$ if the deal commits,
and $\OwnsA(P,a)$ to indicate that $P$ will own $a$ if the
deal aborts.

Escrow plays the role of classical concurrency control,
ensuring that a single asset cannot be transferred to
different parties at the same time.
Here is what happens when $P$ places $a$ in escrow during deal $D$:
\begin{align*}
\text{Pre:\quad}  & \Owns(P,a) \\
\text{Post:\quad} &\Owns(D,a) \text{ and }
                    \OwnsC(P,a) \text{ and }
                    \OwnsA(P,a)
\end{align*}
The precondition states that $P$ can escrow $A$ only if $P$ owns $a$.
If that precondition is satisfied,
the postcondition states that ownership of $a$ is transferred
from $P$ to $D$ (via the escrow contract),
but $P$ remains the owner of $a$ in both $C$ and $A$,
since no tentative transfer has happened yet,
so $P$ would regain ownership of $a$ if $D$ were to terminate either way.
For example, when Bob escrows his tickets,
they become the property of the contract,
but should the deal terminate right then,
the tickets would revert to Bob.

Next we define what happens when party $P$ tentatively transfers an asset
(or assets) $a$
to party $Q$ as part of deal $D$.
\begin{align*}
\text{Pre:\quad}  &\Owns(D,a) \text{ and } \OwnsC(P,a)\\
\text{Post:\quad} &\OwnsC(Q,a)
\end{align*}
The precondition requires $a$ to be held in escrow by $D$,
with $P$ the indicated owner should $D$ commit.
If the precondition is satisfied,
the postcondition states that $Q$ will become the owner of the
transferred $a$ should $D$ commit
(at this point).
For example,
when Carol transfers 101 coins to Alice,
Alice becomes the owner of those coins in $C$.
Alice can then transfer 100 of those coins to Bob,
retaining one for herself, all in $C$.

Assets remain in escrow until the deal terminates.
If the deal terminates by committing,
the owners of assets in $C$ become the actual owners
(displacing $D$).
If it terminates by aborting,
the owners of assets in $A$ become the actual owners 
(again displacing $D$).

\subsection{Phases}
A deal is executed in the following phases.

\paragraph*{Clearing Phase}
  A market-clearing service discovers and broadcasts the participants,
  the proposed transfers, and possibly other deal-specific information.
The market clearing service may be centralized,
but \emph{it is not a trusted party},
because each party later decides for itself whether to participate.
The precise structure of the service is beyond the scope of this paper.
  
\paragraph*{Escrow Phase}
  Parties escrow their outgoing assets.
  For example, Bob escrows his tickets and Carol her coins.

\paragraph*{Transfer Phase}
  The parties perform the sequence of tentative ownership transfers according to
  the deal.
  For example,
  Bob tentatively transfers the tickets to Alice, who subsequently
  transfers them to Carol.

\paragraph*{Validation Phase}
  Once the tentative transfers are complete,
  each party checks that the deal is the same as proposed by the (untrusted)
  market-clearing service,
 that its incoming assets are properly escrowed
  (so they cannot be double-spent),
  and that the payoff defined by the incoming and outgoing assets is acceptable.
  For example, Carol checks that the tickets to be transferred are escrowed,
  that the seats are (at least as good as) the ones agreed upon,
  and that she is not about to somehow overpay.

In the classical two-phase commit protocol~\cite{BernsteinHZ1986},
validation usually requires no semantic checks;
instead a party agrees to prepare if appropriate locks are held and
persistence is guaranteed.
Under adversarial commerce, however,
an application-specific validation phase is needed for each party to
decide whether the proposed payoff is acceptable.
For example,
only Carol can decide whether the tickets she is about to purchase
are ones she wants.

\paragraph*{Commit Phase}
  The parties vote on whether to make the tentative transfers permanent.
  If all parties vote to commit,
  the escrowed assets are transferred to their new owners; otherwise
they are refunded to their original owners.

Cross-chain deals rely on two critical, intertwined mechanisms.
First,
the escrow mechanism prevents double-spending by making the
escrow contract itself the asset owner.
Care must be taken that assets belonging to compliant parties do not
remain escrowed forever in the presence of malicious behavior by
counterparties.
Second,
the commit protocol must be resilient in the presence of malicious
misbehavior.
A deviating party may be able to steal assets if it can convince
some parties that the deal completed successfully,
and others that it did not.
If a deviating party can prevent (or delay) a decision by the commit protocol,
then it can keep assets locked up forever (or a long time).

The principal challenge in implementing cross-chain deal protocols is
the design of the integrated escrow management and commit protocol.
Just as with classical transaction mechanisms,
there are many possible choices and trade-offs.
In the remainder of this paper,
we describe two cross-chain deal protocols, implemented via contracts,
one for the synchronous timing model,
and one for the semi-synchronous model,
each making different trade-offs concerning decentralization and fault-tolerance.

%% file: timelock.tex
We now describe a \emph{timelock} commit protocol where
escrowed assets are released if all parties vote to commit.
Parties do not explicitly vote to abort.
Instead, timeouts are used to ensure that escrowed assets are not
locked up forever if some party crashes or walks away from the deal.
This protocol assumes a \emph{synchronous} network model where
blockchain propagation time is known and bounded.
 
In our example, Bob places his tickets into escrow, then transfers
them to Alice, who transfers them to Carol.  All parties examine their
incoming assets, and if the resulting payoff is acceptable,
the parties
vote to commit at the escrow contract on each asset's blockchain.  For
example, if Alice, Bob, and Carol all register commit votes on the
ticket blockchain, the escrow contract releases the tickets to Carol.
All votes are subject to timeouts: if any commit vote fails to appear
before the contract's timeout expires, the tickets revert to Bob.
(Symmetric conditions apply to Carol's coins.)

Because of the  adversarial nature of a deal,
each party is motivated to publish its vote on the blockchains
controlling its incoming assets (it is eager to be paid),
but not on the blockchains controlling its outgoing assets
(it is not so eager to pay).
To align the protocol with incentives,
one party's commit vote may be \emph{forwarded} from one escrow
contract to another by a motivated party.

For example,
Bob is motivated to publish his commit vote only on the coin blockchain.
However, once published, Bob's vote becomes visible to Carol,
who is motivated to forward  that vote to the ticket blockchain.
Carol's position is symmetric:
she is motivated to publish her vote only on the ticket blockchain,
but Bob is motivated to forward it to the coin blockchain.
Alice is motivated to send her vote to both blockchains.
(Nevertheless, no harm occurs if a party sends its commit vote
directly to any contract.)

A tricky part of this protocol is how to choose timeouts.
A protocol implementation that simply assigns each party a timeout for
each asset does not satisfy our notions of correctness,
as shown by the following example.

Suppose that the ticket and coin escrows assign Alice timeouts $A_t$
and $A_c$ respectively, and that Bob and Carol's commit votes have
already been published on both blockchains.  In one scenario, Alice
waits until just before $A_c$ to register her vote on the coin
blockchain, unlocking Carol's payment to Bob.  It may take time
$\Delta$ for Carol to observe Alice's vote and forward it to the ticket
blockchain, implying that $A_t \geq A_c+\Delta$.  In another scenario,
Alice waits until just before $A_t$ to register her vote on the ticket
blockchain, unlocking Bob's tickets for Carol.  It may take time
$\Delta$ for Bob to observe Alice's vote and forward it to the coin
blockchain, implying that $A_c \geq A_t+\Delta$, a contradiction.

To resolve this dilemma,
each escrow contract's timeout for a party's commit vote depends on
the length of the path along which that vote was forwarded.
For example,
if Alice votes directly,
her vote will be accepted only if it is received within $\Delta$
of the commit protocol's starting time.
This vote must be signed by Alice.
If Alice forwards a vote from Bob,
that vote will be accepted only if it is received within $2 \cdot \Delta$ of
the starting time,
where the extra $\Delta$ reflects the worst-case extra time needed to
forward the vote.
This vote must be signed first by Bob, then Alice.
Finally,
if Alice forwards a vote that Bob forwarded from Carol,
that vote will be accepted only if it is received within
$ 3 \cdot \Delta$, and so on.
This vote must be signed first by Carol, then Bob, then Alice.
We refer to this chain of signatures as the vote's \emph{path signature}.

In general, a vote from party $X$ received with path signature $p$ must
arrive within time $|p| \cdot \Delta$ of the pre-established commit
protocol starting time, where $|p|$ is the number of distinct
signatures for that vote.

\subsection{Running the Protocol}
Here is how to execute the phases of a timelock protocol.

\paragraph*{ Clearing Phase}
The market-clearing service broadcasts the following to all parties in
the deal: the deal identifier $D$, the list of parties $\plist$,
a commit phase starting time $t_0$ used to compute timeouts,
and the timeout delay $\Delta$.
Most blockchains measure time imprecisely,
usually by multiplying the current block height by the average block rate.
The choice of $t_0$ should be far enough in the future to take into
account the time needed to perform the deal's tentative transfers,
and $\Delta$ should be large enough to render irrelevant any imprecision in
blockchain timekeeping.
Because $t_0$ and $\Delta$ are used only to compute timeouts,
their values do not affect normal execution times,
where all votes are received in a timely way.
If deals take minutes (or hours), then $\Delta$ could be measured in hours (or days).

\paragraph*{ Escrow Phase}
Each party places its outgoing assets in escrow
through an escrow contract
\begin{equation*}
\escrow(D, \Dinfo, a).
\end{equation*}
on that asset's blockchain.  Here $D$ is the deal identifier
and $\Dinfo$ is the rest of the information about the deal ($\plist$,
$t_0$, and $\Delta$); the escrow requests takes effect only if the
party is the owner of $a$ and a member of the $\plist$.

\paragraph*{ Transfer Phase}
Party $P$ transfers an asset (or assets) $a$
tentatively owned by $P$ to party $Q$
by sending
\begin{equation*}
\transfer(D, a, Q).
\end{equation*}
to the escrow contract on the asset's blockchain.
The party must be the owner of $a$ and $Q$ must be in the $\plist$.

\paragraph*{ Validation Phase}
Each party examines its escrowed incoming assets to see if they
represent an acceptable payoff and the deal information
provided by the market-clearing service  is correct.
If so, the party votes to commit.

\paragraph*{ Commit Phase}
Each compliant party sends a commit vote to the escrow contract for each incoming asset.
(A compliant party is free to altruistically send commit votes to
other escrow contracts as well.)
A party uses
\begin{equation*}
\commit(D, v, p)
\end{equation*}
to vote directly and to forward votes
to the deal's escrow contracts,
where $v$ is the voter and $p$ is the path signature for $v$'s vote.
For example, if Alice is forwarding Bob's vote then $v$ is Bob,
and $p$ contains first Bob's signature, and then Alice's signature.
(Throughout, we assume that deal identifiers are unique to guard
against replay attacks.) 

A contract accepts a commit vote only if it arrives in time and is
well-formed: all parties in the path signature are unique and in the
$\plist$, and their signatures are valid and attest to a vote from $v$.
If the commit is accepted, that contract has now
accepted a vote from the party.

A contract releases the escrowed asset to the new owner(s) when it
accepts a commit vote from every party.  If the contract has not
accepted a vote from every party by time $t_0 + N \cdot \Delta$, where
$N$ is the number of parties, it will never accept the missing
votes, so the contract times out and refunds its escrowed assets to
the original owners.

\subsection{Well-formed Deals and Decentralization}
\label{sec:decent}

\begin{figure}
\centering
  \includegraphics[width=0.8 \hsize]{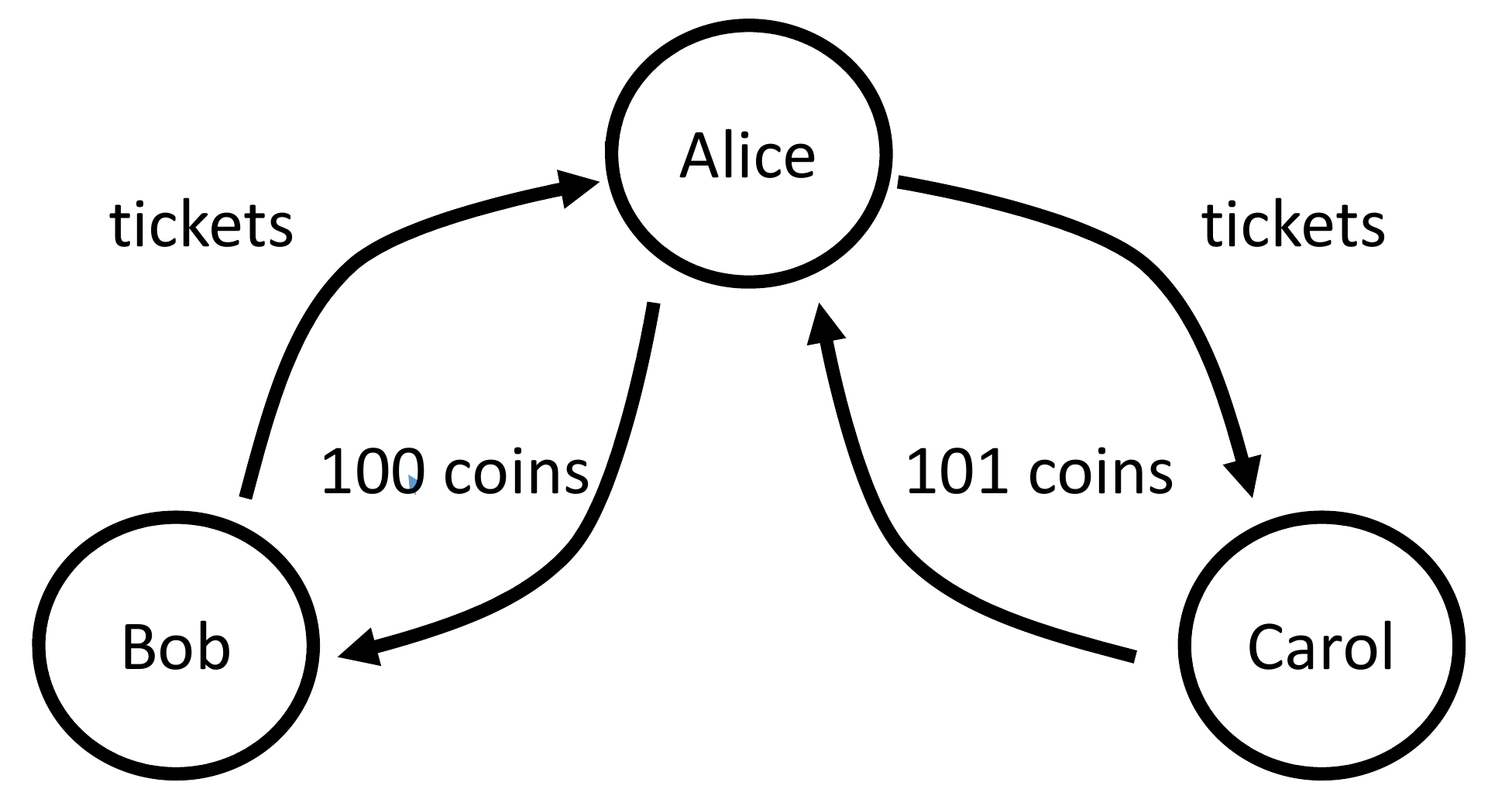}
  \caption{Alice, Bob, and Carol's deal expressed as a digraph}
  \figlabel{digraph}
\end{figure}

For ease of exposition, we can think of a deal as a
\emph{directed graph} (digraph),
where each vertex represents a party, and each arc represents a
transfer; the digraph for our deal is shown in \figref{digraph}.

If the deal digraph is not strongly connected, it can be shown
that the deal is not \emph{well-formed}, in the sense that it must
include one or more ``free riders'' that collectively take assets but
do not return any~\cite{Herlihy2018}.  The remaining parties have no
incentive to conform to any protocol executing such a deal, because
they could improve their payoffs  by excluding the free
riders\footnote{ Perhaps the free riders are sending some kind of
  hidden off-chain payments to the other parties,
  but support for hidden payments is beyond this paper's scope.}.

A compliant party first sends votes to the escrow contracts on  its incoming
assets' blockchains.
Then it monitors its outgoing assets' blockchains and forwards
other parties' votes to its incoming assets' blockchains.
No party needs to interact with any other blockchains.
For example, if Carol owns only altcoins,
then as part of the deal,
she can go to David to exchange her altcoins for coins,
and the deal can commit without parties such as Bob needing to
interact with the altcoin blockchain (or even know about it).
This protocol is \emph{decentralized} in the sense that
there is no single blockchain that must be accessed by all
compliant parties.

This voting protocol reflects the
(incentive-compatible) \emph{minimum} a compliant party must do.
Nothing prevents compliant parties from sending their commit votes
directly to arbitrary blockchains, although typically parties
will want to restrict themselves to blockchains they already use.
For example,
Bob might send his vote directly to the ticket blockchain,
perhaps hoping to speed up the commit process.
If he does so, he passes up a (very unlikely) opportunity to cheat Carol
if she (non-compliantly) fails to claim her tickets in time.

In the remainder of this section, we assume that
deals are well-formed, and the corresponding digraphs are strongly
connected, although the timelock protocol can handle ill-formed deals
if needed.

\subsection{What Could Possibly Go Wrong?}
What happens if parties choose to deviate?
For example,
suppose Bob wants to trade b-coins for c-coins,
and Carol wants the reverse.
Alice brokers their deal,
taking 101 b-coins (c-coins) from Bob (Carol),
then forwarding 100 coins to each counterparty,
keeping a 1-coin commission from each side.
This time, Alice happens to hold both kinds of coins.
To save time,
she transfers 101 of Bob's b-coins into her account,
and simultaneously transfers 100 of her own b-coins to Carol,
and similarly for Carol (in the opposite direction).
Now suppose Alice is infected by a virus,
and starts to behave irrationally.
After Bob releases his vote,
Alice stops communicating with him,
but she communicates normally with Carol.
Bob's timelock will eventually expire,
and he will take back his coins, so for him the deal aborted.
Carol, however, will transfer her 101 c-coins to Alice and receive her
100 b-coins, so for her the deal committed normally.

Although this outcome is not ``all-or-nothing'',
it is considered acceptable because all compliant parties
end up with acceptable outcomes,
But Alice, who deviated from the protocol, foots the bill,
paying Carol without being paid by Bob.
We emphasize the unenforceability of classical correctness properties
because it may seem counter-intuitive.
But if we want to verify that contracts are correct,
we must take care to use a realizable notion of correctness.

\subsection{Correctness}

\begin{theorem}
The timelock protocol satisfies safety.
\end{theorem}

\begin{proof}
By construction,
transferring a compliant party $X$'s escrowed incoming and outgoing
assets is an acceptable payoff for $X$.
Suppose by way of contradiction that
$X$'s outgoing asset $a$ is released from escrow and transferred
(with commit votes from every party),
but the escrow for $X$'s incoming asset $b$ times out and is refunded
because of a missing vote from party $Z$.
Suppose $Z$'s commit vote at $a$'s contract arrived with path signature $p$.
The signatures in $p$ cannot include $X$'s,
because $X$ is compliant and would have already
forwarded $Z$'s vote to $b$.  $Z$'s vote must have arrived at $a$
before time $t_0 +|p| \cdot \Delta$. Since $X$ is compliant, it
forwards that vote to $b$'s contract before time
$t_0 + (|p|+1) \cdot \Delta$,
where that vote is accepted, a contradiction.
\end{proof}

\begin{theorem}
The timelock protocol satisfies weak liveness:
no compliant party's outgoing assets are locked up forever.
\end{theorem}

\begin{proof}
Every escrow created by a compliant party has a finite timeout.
\end{proof}

\begin{theorem}
The timelock protocol satisfies strong liveness.
\end{theorem}

\begin{proof}
If all parties are compliant,
they all send commit votes to the escrow contracts for
their incoming assets.
Each time a new commit vote appears on an outgoing asset's contract,
the party forwards it to its incoming assets' contracts.
Since the deal is well-formed, the deal digraph is strongly connected,
and all commit votes are forwarded to all contracts in time.
\end{proof}

Suppose that Bob acquires Alice and Carol's votes on time, and
forwards them to claim the coins, but Alice and Carol are driven
offline before they can  forward Bob's vote to the ticket
blockchain, so Bob ends up with both the coins and the
tickets.
Technically, Alice and Carol have deviated from the protocol by not
claiming their assets in time.
As a countermeasure, $\Delta$ should be chosen large enough to make
sustained denial-of-service attacks prohibitively expensive.
For similar reasons, the Lightning payment network~\cite{lightning}
employs \emph{watchtowers}~\cite{watchtower} to monitor escrow
contracts to act on the behalf of off-line parties.

%% file: cbc.tex
Now we describe a commit protocol that assumes only a semi-synchronous communication model~\cite{Dwork:1988},
Since we cannot use timed escrow,
we allow parties to vote to abort if validation fails,
or if too much time has passed.

Unlike in the classical two-phase commit protocol~\cite{BernsteinHZ1986},
there is no coordinator; instead we use a
special blockchain, the \emph{certified blockchain}, or \emph{CBC}, as
a kind of shared log. The CBC might be a stand-alone blockchain or one
of those already being used in the deal.

Instead of voting on individual assets,
each party votes on the CBC whether to commit or abort the entire deal.
The CBC records and orders these votes.
A party can extract a \emph{proof} from the CBC that particular
votes were recorded in a particular order.
A party claiming an asset (or a refund) presents a proof of commit
(or abort) to the contract managing that asset.
The contract checks the proof's validity and carries out the
requested transfers if the proof is valid.
A \emph{proof of commit} proves that every party voted to commit the
deal before any party voted to abort.
A \emph{proof of abort} proves that some party voted to abort
before every party voted to commit.
A party can rescind an earlier commit vote by voting to abort
(for example, if the deal is taking too long to complete).
To ensure strong liveness,
once a compliant party has voted to commit,
it must wait long enough to give the other parties a chance to vote
before it changes its mind and votes to abort.

Recall that a commit protocol is \emph{decentralized} if there is
no single blockchain accessed by all parties in any execution
(see Section~\ref{sec:decent}).
The CBC protocol is not decentralized in this sense,
because the CBC itself is a centralized ``whiteboard'' shared by all
parties.
This loss of decentralization is inevitable:
no protocol that tolerates periods of asynchrony can be decentralized.
A complete formal proof is out
of scope, but we outline an argument adapted from Fischer,
Lynch, and Paterson~\cite{FischerLP1985}.  If all parties are
compliant, then if the deal commits (resp. aborts) at any asset's
blockchain, it must commit (resp. abort) at all of them.  Initially,
the deal's state is \emph{bivalent}: both commit and abort are
possible outcomes.  But the deal's state cannot remain bivalent
forever, so it must be possible to reach a (bivalent) \emph{critical
state} where each party is about to take a decisive step that will
force the protocol to enter a \emph{univalent} state where either a
commit outcome or an abort outcome becomes inevitable.  A potentially
decisive step forcing an eventual commit cannot take place at a
different blockchain than a potentially decisive step forcing an
eventual abort, because then it would be impossible to determine which
happened first, hence which one was truly decisive.  It follows that
in any such critical state,
all parties must be about to call the same contract,
violating decentralization.

\subsection{Running the Protocol}
Here is how to execute the phases of a CBC protocol.

\paragraph*{Clearing Phase}
The market-clearing service broadcasts a unique identifier $D$ and a
list of participating parties $\plist$ (this protocol does not require
the $t_0$ starting time or $\Delta$).
One party records the start of the deal on the CBC by
publishing an entry:
\begin{equation*}
\startD(D, \plist).
\end{equation*}
The calling party must appear in the $\plist$.
If more than one $\startD$ for $D$ is recorded on the CBC,
the earliest is considered definitive.

\paragraph*{Escrow Phase}
Each party places its outgoing assets in escrow:
\begin{equation*}
\escrow(D, \plist, h, a, \ldots)
\end{equation*}
Here, $h$ is the hash that identifies a particular $\startD$ 
entry on the CBC that started the deal; it is needed
in case there is more than one such entry on the CBC.  The
ellipsis indicates arguments that vary depending on the algorithm used
to implement the CBC, as discussed in \secref{proof}.  As in the
timelock protocol, the sender must be the owner of asset $a$ and a
member of $\plist$.

\paragraph*{Transfer Phase}
Party $P$ transfers an asset (or assets) $a$
tentatively owned by $P$ to party $Q$
by sending
\begin{equation*}
\transfer(D, a, Q).
\end{equation*}
to the escrow contract on the asset's blockchain.
$P$ must be the owner of $a$ and $Q$ must be in the $\plist$.

\paragraph*{Validation Phase}
As before,
each party checks that its proposed payoff is acceptable and
that assets are properly escrowed with the correct $\plist$ and $h$.

\paragraph*{Commit Phase}
Each party $X$ publishes either a commit or abort vote for $D$ on the CBC:
\begin{eqnarray*}
\commit(D, h, X) \text{\quad or \quad} \abort(D, h, X)
\end{eqnarray*}
where $D$ is the deal identifier, and $h$ is the $\startDeal$.
As usual,
each voter must be in the start-of-deal $\plist$.

\subsection{What Could Possibly Go Wrong?}
When things go wrong,
the CBC protocol permits fewer outcomes than the timelock protocol,
because all compliant parties agree on whether the deal
committed or aborted.
Nevertheless, because we cannot constrain the behavior of deviating parties,
even this protocol cannot enforce the classical ``all-or-nothing''
property of atomic transactions.
For example,
if (deviating) Carol erroneously sends 1001 coins to Alice,
instead of the 101 expected, and all parties vote to commit,
then (compliant) Alice ends up with a commission of 901 coins,
an outcome that is neither ``all'' nor ``nothing'',
even for compliant parties.
Of course, such an outcome is unlikely in practice,
but such distinctions matter when reasoning about correctness.
Both the timelock and CBC protocols satisfy the (informal) safety
property that no compliant party can end up ``worse off.''

\subsection{Correctness}
The correctness of the CBC protocol is mostly self-evident:
safety is satisfied because complaint parties agree on whether
a deal commits or aborts.
Weak liveness is satisfied because any compliant party whose assets
are locked up for too long will eventually vote to abort,
and strong liveness is satisfied in periods when the network is
synchronous because every party votes to commit before any
party votes to abort.

\subsection{Cross-Chain Proofs}
\seclabel{proof}
It is easy for (active) parties to ascertain whether a deal committed
or aborted; 
it is not so easy for passive contracts,
which cannot directly observe other blockchains, to do so.

A deal's \emph{decisive vote} is the one that determines whether
the deal commits or aborts.
A straightforward approach is to present each contract with a
subsequence of the CBC's blocks,
starting with the deal's first $\startX$ record,
and ending with its decisive vote.
But how can the contract tell whether the blocks presented are really
on the CBC? 
The answer depends partly on the kind of algorithm underlying the CBC
blockchain.

\subsection{Byzantine Fault-Tolerant Consensus}
Let us assume the CBC relies on \emph{Byzantine fault-tolerant} (BFT)
consensus~\cite{AbrahamGM2018,hyperledgerfabric,PBFT,tendermint-github}.
BFT protocols guarantee safety even when communication is
asynchronous,
and they ensure liveness when communication becomes synchronous after the GST.

Blocks are approved by a known set of $3f+1$ \emph{validators},
of which at most $f$ can deviate from the protocol.
(The details of how validators reach consensus on new blocks
are not important here.)
To support long-term fault tolerance, the
blockchain is periodically \emph{reconfigured} by having at least
$2f+1$ current validators elect a new set of validators.
For ease of exposition,
assume each block contains the next block's group of validators and their keys.

Each block in a BFT blockchain is vouched for by a certificate
containing at least $f+1$ validator signatures of that block's hash.
(Any $f+1$ signatures are enough because at least one of them
must come from an honest validator.)
The sequence of blocks and their certificates can be used as a proof.  The contract on the asset blockchain will be able to
check this proof as long as it knows the first block's set of
validators, accomplished by passing the $3f+1$ initial block validators
as an argument to each of the deal's escrow
contracts (in place of the ellipses).  Parties must identify
correct validators when putting assets in escrow, and they must
check validators' credentials before voting to commit.

Checking the proof as just described is a lot of work; the proof is
likely to be spread over many blocks, each containing a large number
of entries.  Furthermore, we cannot shorten the proof by omitting irrelevant
block entries, because then a malicious party might fool a contract into making a wrong
decision. But there are many ways to make BFT proofs more efficient.

A straightforward optimization is to take advantage of the
fact that the CBC has validators. This allows the parties to
request
certificates from the CBC.  Such a
certificate would vouch for the current state of the deal (active,
committed, aborted).
This certificate alone would constitute a proof provided the
original validators are still active,
and otherwise the party must also provide the chain of validators
across each reconfiguration.

\subsection{Proof-of-work (Nakamoto) Consensus}
Proofs of commit or abort generated by a CBC implemented using
proof-of-work consensus
(like Bitcoin~\cite{bitcoin} or Ethereum~\cite{ethereum})
are possible, but care is needed because such blockchains lack \emph{finality}:
any proof might be contradicted by a later proof,
although forging a later, contradictory proof becomes more expensive to
the adversary the longer it waits.
(Kiayias \emph{et al.}~\cite{KiayiasLS2016,cryptoeprint:2017:963}
propose changes to standard PoW protocols that would make
such ``proofs-of-proof-of-work'' more compact.)

Here is a scenario where Alice can construct a fake ``proof of abort''
for a proof-of-work CBC.  As soon as the deal execution starts, Alice
(perhaps aided by partners in crime) privately mines a block that
contains an $\abort$ vote from Alice.  When her part of the deal is
complete, however, Alice publicly sends a $\commit$ vote to the CBC.
If, by the time all parties have voted $\commit$, Alice was able to
mine a private $\abort$ block, then Alice can use that fake proof of
abort to halt outgoing transfers of her assets, while using the
legitimate proof of commit to trigger incoming transfers.

In the spirit of proof-of-work,
such an attack can be made more expensive by requiring a proof of commit
or abort to include some number of \emph{confirmation} blocks beyond
the one containing the decisive vote,
forcing Alice to outperform the rest of the CBC's miners for an
extended duration.
To deter rational cheaters,
the number of confirmations required should vary depending on
the value of the deal,
implying that high-value deals would take longer to resolve than
lower-value deals.

To summarize, while it is technically possible to produce commit or
abort proofs from a proof-of-work CBC, the result is likely to be slow
and complex.  In the same way a proof-of-work blockchain can fork,
a ``proof-of-proof-of-work''~\cite{KiayiasLS2016} can be contradicted
by a later ``proof-of-proof-of-work''.
Similarly, to make the production of contradictory proofs expensive,
the proof's difficulty must be adjusted to match the value of the
assets transferred by the deal.  By contrast, a BFT certificate of
commit or abort is final,
and independent of the value of the deal's assets.

%% file: costs.tex
This section analyzes the costs associated with each of the protocols.
The code shown here is not intended to be a detailed
implementation; it is only intended to illustrate how such
implementations might be organized.
To compare their implementation costs,
we use a cost model inspired by the Ethereum~\cite{yellowpaper} blockchain,
currently the best-developed platform.
The costs of non-PoW blockchains are likely to be similar.

\tabref{gas} summarizes gas costs for a deal with $n$ parties,
$m$ assets, and $t \geq n$ transfers.

\subsection{Gas Costs}
\begin{table*}[htb]
  \caption{Gas costs}
  \tablabel{gas}
\begin{center}
    \begin{tabular}{| c | c | c | c |}
    \hline
    Protocol & Escrow        & Transfer and Validation & Commit  or Abort\\ \hline
    Timelock & $O(m)$ writes & $O(t)$ writes           &$O(m n^2)$ sig. ver. $+$ $O(m)$ writes \\ \hline
    CBC      & $O(m)$ writes & $O(t)$ writes           &$O(m(f+1))$ sig. ver. $+$ $O(m)$ writes\\ \hline
    \hline
    \end{tabular}
\end{center}
\end{table*}
\begin{figure*}[htb]
\centering
  \begin{lstlisting}
contract EscrowManager {
    ERC20Interface asset;            // contract holding assets
    mapping(address => uint) escrow; // escrowed assets`\linelabel{escrow}`
    mapping(address => uint) onCommit; // result of tentative transfers`\linelabel{onCommit}`
    ...
    // transfer into escrow account
    function escrow (uint amount) public {
        require (asset.transferFrom(msg.sender, this, amount));`\linelabel{xferfrom}`
        escrow[msg.sender] = escrow[msg.sender] + amount;`\linelabel{eupdate}`
        onCommit[msg.sender] = onCommit[msg.sender] + amount;`\linelabel{ocupdate}`
    }
    // tentative transfer
    function transfer (address to, uint amount) public {
        require (onCommit[msg.sender] >= amount);
        onCommit[msg.sender] = onCommit[msg.sender] - amount;`\linelabel{ocupdate2}`
        onCommit[to] = onCommit[to] + amount;`\linelabel{ocupdate3}`
    } ... }
  \end{lstlisting}
  \caption{Pseudocode code for Escrow and Transfer (some details omitted)}
  \figlabel{escrowtransfer}
\end{figure*}
To make denial-of-service attacks prohibitively expensive,
virtual machines that execute contracts typically charge for each instruction executed.
In Ethereum~\cite{yellowpaper},
this charge is expressed in terms of \emph{gas} units,
whose value (in Ether, the blockchain's native currency) varies
according to demand.
For example,
the gas cost of simple arithmetic operations or accesses to short-lived memory is in single digits,
and control flow or read operations from long-lived storage is in double or triple digits.
In general, gas costs are dominated by two kinds of operations:
writing to long-lived storage is (usually) 5000 gas,
and each signature verification is 3000 gas.

\sloppy
To illustrate our gas cost analysis,
\Figref{escrowtransfer} shows a fragment of a pseudocode
implementation of a generic \var{EscrowManager} contract for a fungible
asset, modeled as an ERC20-standard token~\cite{erc20}.
The heart of the \var{EscrowManager} contract is a pair of mappings:
\var{escrow} records how many tokens each party has escrowed
(\lineref{escrow}),
and \var{onCommit} records how many tokens each party would receive if
the deal commits (\lineref{onCommit}).
For clarity, some error checking has been omitted.

\paragraph*{ Escrow Phase}
Each party calls the $escrow$ function to escrow some number of
tokens.
This function incurs 2 storage writes (in a function call) to transfer the token
from the sender to the escrow contract (\lineref{xferfrom}),
and 1 storage write each to update the \var{escrow} (\lineref{eupdate})
and the \var{onCommit} (\lineref{ocupdate}) maps,
for a total of 4 storage writes.
Globally, the escrow phase incurs $O(m)$ gas costs.

\paragraph*{ Transfer Phase}
Each party calls the \var{transfer} function to transfer some number of
escrowed tokens to another party.
This function incurs 1 storage write to decrement the sender's
tentative \var{onCommit} balance (\lineref{ocupdate2}),
and another to increment the recipient's balance (\lineref{ocupdate3}).
Globally, the transfer phase incurs $O(t)$ gas costs.

\paragraph*{ Validation  Phase}
Each party monitors its incoming and outgoing escrow contracts to ensure
it is satisfied with the assets it is due to acquire and relinquish.
This computation takes place entirely at the parties,
and incurs no gas cost.

\begin{figure*}[htb]
  \centering
  \begin{lstlisting}
contract TimelockManager is EscrowManager{
    address[] parties;		// participating parties`\linelabel{parties}`
    address[] voted;		// which parties have voted`\linelabel{whovoted}`
    ...
    function commit (address voter, address[] signers, bytes32[] sigs) public {
	require (now < start + (path.length() * DELTA)); // not timed out`\linelabel{timeout}`
	require (parties.contains(voter));		 // legit voters only`\linelabel{legit}`
	require (!voted.contains(voter));		 // no duplicate votes`\linelabel{dup}`
	require (checkUnique(signers));			 // no duplicate signers`\linelabel{unique}`
	for (int i = 0 ; i < signers.length; i++) {
	    require (checkSig(voter, signers[i], sigs[i])); // expensive`\linelabel{verify}`
	}
	voted.push(voter);                               // remember who voted`\linelabel{voted}`         
    }}
  \end{lstlisting}
  \caption{Pseudocode Fragment for Timelock contract voting (some details omitted)}
  \figlabel{timelock}
\end{figure*}

\paragraph*{ Timelock Protocol Commit Phase}
Timelock escrow contracts verify commit path signatures.
Note that signatures are generated by parties,
not by contracts,
so while signature generation incurs computation costs at parties,
it incurs no gas costs at contracts.

\Figref{timelock} shows a pseudocode fragment for a timelock escrow contract.
The contract records the set of parties participating in the deal
(\lineref{parties}) and which ones have voted (\lineref{parties}).
The $commit$ function takes as arguments the
voter, the set of signers, and their signatures.  It checks that the
deal has not timed out (\lineref{timeout}), that the voter is
legitimate (\lineref{legit}), that the vote has not already been
recorded (\lineref{dup}), and that there are no duplicate signers
(\lineref{unique}).  The expensive steps are verifying each of the
signatures (\lineref{verify}), and recording the voter
(\lineref{voted}) in long-lived storage.

Each escrow contract verifies a vote from each of $n$ parties,
and each party's vote could have been signed by up to $n-1$ others,
yielding a worst-case per-contract bound of $O(n^2)$ signature verifications,
plus a constant number of storage writes for other bookkeeping.
Since there are $m$ contracts,
the timelock commit protocol incurs an $O(m n^2)$ global gas cost.
In the best case,
a deal can abort with no signature verifications,
but in the worst case,
aborting can cost almost as much as committing.

\paragraph*{ CBC Protocol Commit Phase}
\begin{figure*}[htb]
  \centering
  \begin{lstlisting}
contract CBCManager is EscrowManager{
    address[] validators;	// CBC validators`\linelabel{validators}`
    ...
    // check commit proof is valid
    function commit (address[] signers, bytes32[] sigs) public {
	require (checkUnique(signers));         // no duplicate signers`\linelabel{cbcunique}`
	require (validators.contains(signers)); // only validators voted`\linelabel{cbcvalid}`
	require (signers.length >= f+1);	// enough validators voted`\linelabel{enuf}`
	for (int i = 0 ; i <  f+1; i++) {
	    require (checkSig(signers[i], sigs[i])); // expensive`\linelabel{cbcver}`
	}
        outcome = COMMITTED;                    // remember we committed
    }  ... }
  \end{lstlisting}
  \caption{Pseudocode Fragment for CBC proof-checking (some details omitted)}
  \figlabel{cbcproof}
\end{figure*}
Escrow contracts check proofs from the CBC by verifying
that they are correctly signed by enough validators.
\Figref{cbcproof} shows a pseudocode fragment for an escrow contract
for a CBC using an underlying BFT consensus protocol that tolerates
$f$ Byzantine validators .
We assume the optimization where parties request status certificates
from the CBC;
for brevity, we assume there have been no reconfigurations.

The contract keeps track of the CBC's current set of validators
(\lineref{validators}); it also knows their public keys.  The $commit$
function takes as arguments the set of signers and their signatures.
It checks that there are no duplicate validators
(\lineref{cbcunique}), that all signers are validators
(\lineref{cbcvalid}), and that there are enough votes
(\lineref{enuf}).  The expensive step is verifying each of the
validator signatures (\lineref{cbcver}); there will also be a constant
number of storage writes to record the outcome and to update the
escrow and ownership mappings.

Each contract verifies $f+1$ signatures,
or $(k+1)(f+1)$ if the set of validators has changed $k$ times.
The global gas cost is $O(m (f+1))$ signature verifications plus
a constant number of storage writes to update the escrow mappings.

\subsection{Time Costs}
\begin{table*}[htb]
  \caption{Delays for synchronous communication}
  \tablabel{time}
\begin{center}
    \begin{tabular}{| c | c | c | c | c |}
    \hline
    Protocol & Escrow   & Transfer and Validation & Commit      & Abort             \\ \hline
    Timelock & $\Delta$ & $k \Delta$ or $\Delta$  &$O(n) \Delta$& $O(n) \Delta$     \\ \hline
    CBC      & $\Delta$ & $k \Delta$ or $\Delta$  &$O(1) \Delta$& per-party timeout \\ \hline
    \hline
    \end{tabular}
\end{center}
\end{table*}
We analyze each commit protocol's timing delays when the network is synchronous,
with bound $\Delta$ on the time needed both to change a blockchain state,
and to have that change observed by any interested party.
The results are summarized in \tabref{time};
here we assume each asset is transferred $k$ times.

For both protocols,
if all parties are conforming,
the escrow phase takes time at most $\Delta$,
since every party updates its outgoing assets' escrow contracts in
parallel.
Similarly,
the transfer phase takes time at most $k \cdot \Delta$.
It may be possible to execute transfers concurrently,
in which case this phase takes time at most $\Delta$.
At the end of the transfer phase,
validation is local and immediate.

\paragraph*{ Timelock Protocol Commit Phase}
If each party sends its commit vote only to the
blockchains managing its incoming assets,
then the worst-case duration of the commit phase is proportional to the longest sequence of transfers,
which is bounded by $n \Delta$.

\paragraph*{CBC Protocol Commit Phase}
All conforming parties send their votes to the CBC in parallel,
and these votes are available very quickly when the CBC is implemented
using a BFT protocol such as from the PBFT family~\cite{AbrahamGM2018,PBFT}.
It requires at most another $\Delta$ for the escrow contracts to
transfer or refund their assets.

%% file: related.tex
As noted,
in a \emph{cross-chain swap}~\cite{bitcoinwiki,bip199,decred,Herlihy2018,tiersnolan,barterdex,ZakharyAE2019,Catalyst},
each party transfers an asset to another party and halts.  Cross-chain
swaps are attractive because they reduce or eliminate the use of
exchanges, some of which have proved to be
untrustworthy~\cite{mtgox,quadrigacx}.
However, we have seen that cross-chain swaps lack the power to express the simple
brokerage deal described in our example,
as well as auctions and other conventional financial transactions.

To our knowledge,
the only cross-chain swap protocols used in practice are
\emph{hashed timelocked contracts}~\cite{bitcoinwiki,bip199,decred,tiersnolan,barterdex}.
Herlihy~\cite{Herlihy2018} generalizes prior two-party cross-chain swap protocols to
a protocol for multi-party swaps on arbitrary strongly-connected directed graph.
Herlihy also observes that the classical ``all-or-nothing'' correctness property
is ill-suited to cross-chain swaps,
and proposes an alternative correctness property which is more specialized than the one presented here
because it is formulated explicitly in terms of direct swaps,
not the more general structures permitted by cross-chain deals.
For example,
Herlihy assumed that any swap outcome where a party receives only partial
inputs and partial outputs is unacceptable,
but the notions of correctness introduced here allow parties to specify
whether some such partial deal outcomes are acceptable.

The timelock commit protocol presented here has a simpler structure
than the one proposed by Herlihy. That protocol used secrets held by
a carefully-chosen subset of parties. Our protocol replaces
secrets with votes performed by
everyone, so it is possible to treat all parties uniformly, and there
is no need for a careful contract deployment phase.  Our protocol also
clarifies when parties review the transactions' final outcomes.  Both
commit protocols use timeout mechanisms based on path signatures.

Zakhary \emph{et al.}~\cite{ZakharyAE2019} propose a cross-chain
commitment protocol that does not use hashed timelocks.
Instead, participating blockchains exchange ``proofs'' of state changes,
somewhat similar to our CBC proposal,
but because parties need to register their intended transfers at the
start,
this protocol supports only unconditional swaps, not full-fledged deals.

Off-chain payment networks~\cite{DeckerW2015,bolt,arwen2019,raiden,
lightning} and state channels~\cite{counterfactual} use hashed
timelock contracts to circumvent the scalability limits of existing
blockchains.
They conduct repeated off-chain transactions,
finalizing their net transactions in a single on-chain transaction. 
The use of hashed timelock contracts ensures that parties cannot be
cheated if one party tries to settle an incorrect final state.
Lind \emph{et al.}~\cite{teechain} propose using (hardware) trusted
execution environments to ease synchrony requirements.
It remains to be seen whether off-chain networks
can be applied to cross-chain deals.
Arwen~\cite{arwen2019} supports multiple off-chain atomic swaps
between parties and exchanges,
but their protocol is specialized to currency trading and does not
seem to support non-fungible assets.  
Komodo~\cite{barterdex} supports off-chain cross-platform payments.

Sharded blockchains~\cite{chainspace,omniledger}  address
 scalability limits of blockchains by partitioning the state into
 multiple shards so that transactions on different shards can proceed
 in parallel, and support multi-step atomic transactions spanning
 multiple shards. 
 An atomic transaction that spans multiple shards is executed at the
 client in Chainspace~\cite{chainspace} ,
 or  at the server in Omniledger~\cite{omniledger}.
In these systems
a transaction  represents a single trusted party and there is no support for transactions involving untrusted parties.

Chainspace~\cite{chainspace}  allows transactions to specify immutable
 proof contracts to be executed at the server.
 The proofs are used to validate client execution traces resembling optimistic concurrency control.
Channels~\cite{AndroulakiCDK2018},
an extension of Omniledger Atomix protocols,
uses proofs in a two-phase protocol similar to our CBC,
for atomic untrusted cross-shard single-step multi-party UTXO~\cite{utxo} transfers,
but does not support multi-step deals or non-fungible assets.

The BAR (byzantine, altruistic and rational) computation
model~\cite{AiyerACDMP2005, BAR-Primer:2006} supports cooperative services
spanning autonomous administrative domains that are resilient to
Byzantine and rational manipulations.
Like Byzantine fault-tolerant systems,
BAR-tolerant systems assume a bounded number of Byzantine faults,
and as such do not fit the adversarial deal model,
where any number of parties may be Byzantine.

The CBC somewhat resembles an \emph{oracle}~\cite{Augur},
a trusted data feed that reports physical-world occurrences to contracts. 

An early precursor of adversarial commerce was the study of
\emph{federated databases}~\cite{Sheth:1990},
which addressed the problem of coordinating and committing
transactions that span multiple autonomous,
mutually untrusting, heterogeneous data stores.
(Federated databases did not attempt to tolerate arbitrary Byzantine behavior.)

%% file: discussion.tex
\seclabel{discussion}
Deals can provide incentives for good behavior.
For example,
a party might escrow a small deposit that is lost if that party is the first to
cause the deal to fail.
Designing and implementing such incentives is an area of ongoing research~\cite{mechanism}.

In the timelock commit protocol, if $\Delta$ is too small,
parties may be vulnerable to an extended denial-of-service attack,
which can cause them to lose their incoming assets.
There is a similar threat to the CBC commit protocol,
where the CBC itself might be the target of a denial of service attack,
but the effect is different: instead the deal's assets are locked up
for the duration of the attack, not lost forever.

A more subtle issue concerning the CBC commit protocol
is that the parties must trust the CBC not to engage in
\emph{censorship},
where CBC validators selectively choose to ignore certain deals,
causing them to abort when they could otherwise have committed.

The main difference in performance in the timelock and CBC protocols is
the number of signatures that must be verified to complete the protocol,
so we need to consider $n$ parties vs. $f+1$ validators.
Many deals will likely have only a few participants, and in this case
it will be more expensive to commit a CBC deal ($O(m(f+1))$)
than a timelock deal ($O(mn^2)$); in a deal with many participants
the reverse may be true.  Even if the CBC protocol is more
expensive,  one gets what one pays for:
the CBC protocol works in a more demanding timing model.

%% file: conclusions.tex
Today's distributed data management systems face a new and daunting challenge:
enabling commerce among autonomous parties who do not know or trust
one another, a model we have called \emph{adversarial commerce}.
Prior approaches,
such as distributed atomic transactions or cross-chain swaps,
are inappropriate for the trust model or have limited expressive power.

In classical distributed systems,
participating parties are trusted.
Techniques such as replica groups can mask crashes and Byzantine failures,
ensuring that trust is warranted.
Computations are often organized as \emph{atomic transactions} that ensure,
for example, that transactions appear to run in one-at-a-time order,
taking effect either everywhere or nowhere.

By contrast, in adversarial commerce
each party decides for itself whether
to participate in a particular interaction.
Parties agree to follow a common protocol,
an agreement that can be monitored, but not enforced.
Correctness is local and selfish:
all parties that follow the protocol end up ``no worse off'' than when they started,
even in the presence of faulty or malicious behavior by an arbitrary number of other parties.
Examples of adversarial commerce include securities trading,
digital asset management,
the Internet of Things,
supply chain management,
and, of course, cryptocurrencies.

It is easy to confuse cross-chain deals with atomic transactions,
and with cross-chain swaps,
since they address similar needs.
We hope this paper has clarified the critical
distinctions between them, and explained why cross-chain deals are needed to
address the demands of adversarial commerce.

%% file: ms.bbl
\begin{thebibliography}{10}

\bibitem{AbrahamGM2018}
I.~Abraham, G.~Gueta, and D.~Malkhi.
\newblock Hot-stuff the linear, optimal-resilience, one-message {BFT} devil.
\newblock {\em CoRR}, abs/1803.05069, 2018.
\newblock To appear in PODC 2019.

\bibitem{AiyerACDMP2005}
A.~S. Aiyer, L.~Alvisi, A.~Clement, M.~Dahlin, J.-P. Martin, and C.~Porth.
\newblock Bar fault tolerance for cooperative services.
\newblock In {\em Proceedings of the Twentieth ACM Symposium on Operating
  Systems Principles}, SOSP '05, pages 45--58, New York, NY, USA, 2005. ACM.

\bibitem{chainspace}
M.~Al{-}Bassam, A.~Sonnino, S.~Bano, D.~Hrycyszyn, and G.~Danezis.
\newblock Chainspace: {A} sharded smart contracts platform.
\newblock {\em CoRR}, abs/1708.03778, 2017.

\bibitem{hyperledgerfabric}
E.~Androulaki, A.~Barger, V.~Bortnikov, C.~Cachin, K.~Christidis, A.~De~Caro,
  D.~Enyeart, C.~Ferris, G.~Laventman, Y.~Manevich, S.~Muralidharan, C.~Murthy,
  B.~Nguyen, M.~Sethi, G.~Singh, K.~Smith, A.~Sorniotti, C.~Stathakopoulou,
  M.~Vukoli\'{c}, S.~W. Cocco, and J.~Yellick.
\newblock Hyperledger fabric: A distributed operating system for permissioned
  blockchains.
\newblock In {\em Proceedings of the Thirteenth EuroSys Conference}, EuroSys
  '18, pages 30:1--30:15, New York, NY, USA, 2018. ACM.

\bibitem{AndroulakiCDK2018}
E.~Androulaki, C.~Cachin, A.~D. Caro, and E.~Kokoris-Kogias.
\newblock Channels: Horizontal scaling and confidentiality on permissioned
  blockchains.
\newblock In {\em ESORICS}, 2018.

\bibitem{libra}
L.~Association.
\newblock An introduction to libra.
\newblock
  \url{https://libra.org/en-US/wp-content/uploads/sites/23/2019/06/LibraWhitePaper_en_US.pdf},
  2019.
\newblock As of 24 September 2019.

\bibitem{binance}
B.~Barrett.
\newblock Hack brief: Hackers stole \$40 million from binance cryptocurrency
  exchange.
\newblock
  \url{https://www.wired.com/story/hack-binance-cryptocurrency-exchange/}, May
  2019.
\newblock As of 8 July 2019.

\bibitem{BernsteinHZ1986}
P.~A. Bernstein, V.~Hadzilacos, and N.~Goodman.
\newblock {\em Concurrency Control and Recovery in Database Systems}.
\newblock Addison-Wesley Longman Publishing Co., Inc., Boston, MA, USA, 1986.

\bibitem{bitcoinwiki}
bitcoinwiki.
\newblock Atomic cross-chain trading.
\newblock \url{https://en.bitcoin.it/wiki/Atomic_cross-chain_trading}.
\newblock As of 9 January 2018.

\bibitem{bip199}
S.~Bowe and D.~Hopwood.
\newblock Hashed time-locked contract transactions.
\newblock \url{https://github.com/bitcoin/bips/blob/master/bip-0199.mediawiki}.
\newblock As of 9 January 2018.

\bibitem{PBFT}
M.~Castro and B.~Liskov.
\newblock Practical byzantine fault tolerance.
\newblock In {\em Proceedings of the Third Symposium on Operating Systems
  Design and Implementation}, OSDI '99, pages 173--186, Berkeley, CA, USA,
  1999. USENIX Association.

\bibitem{watchtower}
J.~Chester.
\newblock Your guide on bitcoin's lightning network: The opportunities and the
  issues.
\newblock
  \url{https://www.forbes.com/sites/jonathanchester/2018/06/18/your-guide-on-the-lightning-network-the-opportunities-and-the-issues/#6c8d8c0f3677}N,
  June 2018.
\newblock As of 11 December 2018.

\bibitem{BAR-Primer:2006}
A.~Clement, H.~Li, J.~Napper, J.~P.~M. Martin, L.~Alvisi, and M.~Dahlin.
\newblock Bar primer.
\newblock In {\em Proceedings of the International Conference on Dependable
  Systems and Networks (DSN), DCC Symposium}, 2008.

\bibitem{counterfactual}
J.~Coleman, L.~Horne, and L.~Xuanji.
\newblock Counterfactual: Generalized state channels.
\newblock \url{http://l4.ventures/papers/statechannels.pdf}, 2018.

\bibitem{DeckerW2015}
C.~Decker and R.~Wattenhofer.
\newblock A fast and scalable payment network with bitcoin duplex micropayment
  channels.
\newblock In A.~Pelc and A.~A. Schwarzmann, editors, {\em Stabilization,
  Safety, and Security of Distributed Systems}, pages 3--18, Cham, 2015.
  Springer International Publishing.

\bibitem{decred}
DeCred.
\newblock Decred cross-chain atomic swapping.
\newblock \url{https://github.com/decred/atomicswap}.
\newblock As of 8 January 2018.

\bibitem{Dwork:1988}
C.~Dwork, N.~Lynch, and L.~Stockmeyer.
\newblock Consensus in the presence of partial synchrony.
\newblock {\em J. ACM}, 35(2):288--323, Apr. 1988.

\bibitem{ethereum}
Ethereum.
\newblock \url{https://github.com/ethereum/}.

\bibitem{mechanism}
A.~Evans.
\newblock A crash course in mechanism design for cryptoeconomic applications.
\newblock
  \url{https://medium.com/blockchannel/a-crash-course-in-mechanism-design-for-cryptoeconomic-applications-a9f06ab6a976},
  2017.
\newblock As of 11 April 2019.

\bibitem{FischerLP1985}
M.~J. Fischer, N.~A. Lynch, and M.~S. Paterson.
\newblock Impossibility of distributed consensus with one faulty process.
\newblock {\em J. ACM}, 32(2):374--382, Apr. 1985.

\bibitem{FordB2018}
B.~Ford and R.~B\:ohme.
\newblock Rationality is self-defeating in permissionless systems.
\newblock
  \url{https://bford.info/2019/09/23/rational/?fbclid=IwAR3Mh0RsKZm0juma2uXNJXNlwu3Ll7Gls28PnqKVvIDt0o-itNAacNod0Nc},
  Sept. 2019.
\newblock As of 26 September 2019.

\bibitem{erc20}
E.~Foundation.
\newblock Erc20 token standard.
\newblock \url{https://theethereum.wiki/w/index.php/ERC20_Token_Standard},
  2019.
\newblock As of 6 April 2019.

\bibitem{GiladHMVZ2017}
Y.~Gilad, R.~Hemo, S.~Micali, G.~Vlachos, and N.~Zeldovich.
\newblock Algorand: Scaling byzantine agreements for cryptocurrencies.
\newblock In {\em Proceedings of the 26th Symposium on Operating Systems
  Principles}, SOSP '17, pages 51--68, New York, NY, USA, 2017. ACM.

\bibitem{bolt}
M.~Green and I.~Miers.
\newblock Bolt: Anonymous payment channels for decentralized currencies.
\newblock Cryptology ePrint Archive, Report 2016/701, 2016.
\newblock \url{https://eprint.iacr.org/2016/701}.

\bibitem{HaerderR1983}
T.~Haerder and A.~Reuter.
\newblock Principles of transaction-oriented database recovery.
\newblock {\em ACM Comput. Surv.}, 15(4):287--317, Dec. 1983.

\bibitem{arwen2019}
E.~Heilman, S.~Lipmann, and S.~Goldberg.
\newblock The arwen trading protocols.
\newblock \url{https://www.arwen.io/whitepaper.pdf}, Jan. 2019.
\newblock As of 23 February 2019.

\bibitem{Herlihy2018}
M.~Herlihy.
\newblock Atomic cross-chain swaps.
\newblock In {\em Proceedings of the 2018 ACM Symposium on Principles of
  Distributed Computing}, PODC '18, pages 245--254, New York, NY, USA, 2018.
  ACM.

\bibitem{utxo}
Investopedia.
\newblock Utxo.
\newblock \url{https://www.investopedia.com/terms/u/utxo.asp}, 2019.
\newblock As of 7 April 2019.

\bibitem{KiayiasLS2016}
A.~Kiayias, N.~Lamprou, and A.-P. Stouka.
\newblock Proofs of proofs of work with sublinear complexity.
\newblock In {\em International Conference on Financial Cryptography and Data
  Security}, 2016.

\bibitem{cryptoeprint:2017:963}
A.~Kiayias, A.~Miller, and D.~Zindros.
\newblock Non-interactive proofs of proof-of-work.
\newblock Cryptology ePrint Archive, Report 2017/963, 2017.
\newblock \url{https://eprint.iacr.org/2017/963}.

\bibitem{omniledger}
E.~Kokoris~Kogias, P.~S. Jovanovic, L.~Gasser, N.~Gailly, E.~Syta, and B.~A.
  Ford.
\newblock Omniledger: A secure, scale-out, decentralized ledger via sharding.
\newblock In {\em 2018 IEEE Symposium on Security and Privacy (SP)}, page~16,
  2018.

\bibitem{teechain}
J.~Lind, I.~Eyal, F.~Kelbert, O.~Naor, P.~R. Pietzuch, and E.~G. Sirer.
\newblock Teechain: Scalable blockchain payments using trusted execution
  environments.
\newblock {\em CoRR}, abs/1707.05454, 2017.

\bibitem{bitcoin}
S.~Nakamoto.
\newblock Bitcoin: A peer-to-peer electronic cash system.
\newblock May 2009.

\bibitem{raiden}
R.~Network.
\newblock What is the raiden network?
\newblock \url{https://raiden.network/101.html}.
\newblock As of 26 January 2018.

\bibitem{tiersnolan}
T.~Nolan.
\newblock Atomic swaps using cut and choose.
\newblock \url{https://bitcointalk.org/index.php?topic=1364951}.
\newblock As of 9 January 2018.

\bibitem{barterdex}
T.~K. Organization.
\newblock The barterdex whitepaper: A decentralized, open-source cryptocurrency
  exchange, powered by atomic-swap technology.
\newblock
  \url{https://supernet.org/en/technology/whitepapers/BarterDEX-Whitepaper-v0.4.pdf}.
\newblock As of 9 January 2018.

\bibitem{Augur}
J.~Peterson, J.~Krug, M.~Zoltu, A.~K. Williams, and S.~Alexander.
\newblock Augur: a decentralized oracle and prediction market platform.
\newblock \url{https://www.augur.net/whitepaper.pdf}, 2018.
\newblock As of 7 April 2019.

\bibitem{lightning}
J.~Poon and T.~Dryja.
\newblock The bitcoin lightning network: Scalable off-chain instant payments.
\newblock \url{https://lightning.network/lightning-network-paper.pdf}, Jan.
  2016.
\newblock As of 29 December 2017.

\bibitem{Sheth:1990}
A.~P. Sheth and J.~A. Larson.
\newblock Federated database systems for managing distributed, heterogeneous,
  and autonomous databases.
\newblock {\em ACM Comput. Surv.}, 22(3), Sept. 1990.

\bibitem{tendermint-github}
Tendermint.
\newblock \url{http:/https://github.com/tendermint/tendermint/wiki}, Oct 2015.
\newblock Commit \texttt{c318a227}.

\bibitem{twophasecommit}
Wikipedia.
\newblock Two-phase commit protocol.
\newblock \url{https://en.wikipedia.org/wiki/Two-phase_commit_protocol}.
\newblock As of 18 May 2018.

\bibitem{mtgox}
Wikipedia.
\newblock Mt. gox.
\newblock \url{https://en.wikipedia.org/wiki/Mt._Gox}, 2019.
\newblock As of 6 April 2019.

\bibitem{quadrigacx}
Wikipedia.
\newblock Quadriga fintech solutions.
\newblock \url{https://en.wikipedia.org/wiki/Quadriga_Fintech_Solutions}, 2019.
\newblock As of 6 April 2019.

\bibitem{yellowpaper}
G.~Wood.
\newblock Ethereum: A secure decentralised generalised transaction ledger.
\newblock {\em Ethereum project yellow paper}, 151:1--32, 2014.

\bibitem{ZakharyAE2019}
V.~Zakhary, D.~Agrawal, and A.~{El Abbadi}.
\newblock Atomic commitment across blockchains.
\newblock {\em CoRR}, abs/1905.02847, 2019.

\bibitem{Catalyst}
G.~Zyskind, C.~Kisagun, and C.~FromKnecht.
\newblock Enigma catalyst: a machine-based investing platform and
  infrastructure for crypto-assets.
\newblock \url{https://www.enigma.co/enigma_catalyst.pdf}.
\newblock As of 25 January 2018.

\end{thebibliography}
